\documentclass[letterpaper, draftclsnofoot, onecolumn]{IEEEconf}
\IEEEoverridecommandlockouts 

\usepackage{soul}
\usepackage[normalem]{ulem}

\usepackage[latin1]{inputenc}
\usepackage[english]{babel}
\usepackage{amsfonts,amssymb,amsmath,comment}
\usepackage[dvipsnames]{xcolor}
\usepackage{tikz}
\usetikzlibrary{decorations, decorations.text,backgrounds}
\usepackage{graphicx}

\usepackage{cite}
\usepackage{placeins}

\renewcommand{\natural}{{\mathbb{N}}}

\newcommand{\real}{{\mathbb{R}}}

\newcommand{\map}[3]{#1: #2 \rightarrow #3}

\newcommand{\sign}{\operatorname{sign}}

\newcommand{\until}[1]{\{1,\ldots,#1\}}

\newcommand{\EE}{\mathcal{E}} 
\newcommand{\GG}{\mathcal{G}}

\newcommand{\KK}{\mathcal{K}}
\newcommand{\II}{\mathcal{I}} 
\newcommand{\NN}{\mathcal{N}} 
\renewcommand{\SS}{\mathcal{S}}

\newcommand{\VV}{\mathcal{V}}

 \newcommand{\subj}{\text{subj.\ to}}

\newcommand{\argmin}{\mathop{\rm argmin}}

\newcommand{\nbrs}{\mathcal{N}}

\newcommand{\innbrs}{\mathcal{N}}

\usepackage{algorithm}
\usepackage[noend]{algpseudocode}

\makeatletter
\newcommand{\StatexIndent}[1][3]{%
  \setlength\@tempdima{\algorithmicindent}%
  \Statex\hskip\dimexpr#1\@tempdima\relax}
\makeatother

\algnewcommand{\algorithmicgoto}{\textbf{go to }}%
\algnewcommand{\Goto}[1]{\algorithmicgoto Line~\ref{#1}}%
\algnewcommand{\Label}{\State\unskip}

\newtheorem{theorem}{Theorem}[section]
\newtheorem{proposition}[theorem]{Proposition}

\newtheorem{assumption}[theorem]{Assumption}
 
\def\QEDclosed{\mbox{\rule[0pt]{1.3ex}{1.3ex}}} 

\def\QED{\QEDclosed} 

\renewcommand{\lim}{\operatornamewithlimits{lim\vphantom{p}}}

\newcommand\oprocendsymbol{\hbox{$\square$}}
\newcommand\oprocend{\relax\ifmmode\else\unskip\hfill\fi\oprocendsymbol}
\def\eqoprocend{\tag*{$\square$}}

\renewcommand{\theenumi}{(\roman{enumi})}

\graphicspath{{figs/}}

\newcommand{\bx}{\mathbf{x}}

\newcommand{\avg}{\bar{\mathbf{z}}}
\newcommand{\bv}{\mathbf{v}}
\newcommand{\bn}{\mathbf{n}}
\newcommand{\bw}{\mathbf{w}}
\newcommand{\by}{\mathbf{y}}
\newcommand{\bA}{\mathbf{A}}
\newcommand{\bD}{\mathbf{D}}
\newcommand{\bb}{\mathbf{b}}

\newcommand{\bI}{\mathbf{I}}
\newcommand{\1}{\mathbf{1}}
\newcommand{\bpi}{\boldsymbol{\pi}}
\newcommand{\tbpi}{\widetilde{\boldsymbol{\pi}}}
\newcommand{\bphi}{\boldsymbol{\phi}}

\newcommand{\bxi}{\boldsymbol{\xi}}
\newcommand{\kron}{\otimes}
\newcommand{\tf}{\tilde{f}}

\definecolor{blue@O4S}{RGB}{0, 41, 69}
\definecolor{emph@O4S}{RGB}{0, 93, 137}
\definecolor{red@O4S}{RGB}{127,0,0}
\definecolor{gray@O4S}{RGB}{112, 112, 112}

\def \algfullname/{{\sc Block-SONATA}}
\def \algacronym/{{\sc Block-SONATA}}

\begin{document}

\title{Distributed Big-Data Optimization \\via Block-Iterative Convexification and Averaging\vspace{-0.3cm}}

\author{Ivano Notarnicola$^{\ast}$, Ying Sun$^{\ast}$, Gesualdo Scutari, Giuseppe Notarstefano\vspace{-0.3cm}
\thanks{$^\ast$These authors equally contributed and are in alphabetic order.}
\thanks{The work of Notarnicola and Notarstefano has received funding from the
  European Research Council (ERC) under the European Union's
  Horizon 2020 research and innovation programme (grant agreement No 638992 -
  OPT4SMART).
  The work of Sun and Scutari has been supported by the USA National Science
  Foundation under Grants CIF 1632599 and CIF 1719205; and in part 
  by the Office of Naval Research under Grant N00014-16-1-2244.
  }
\thanks{Ivano Notarnicola and Giuseppe Notarstefano are with the
  Department of Engineering, Universit\`a del Salento, Lecce, Italy, 
  \texttt{name.lastname@unisalento.it.}  } 
\thanks{Ying Sun and Gesualdo Scutari are with the School of Industrial Engineering, 
Purdue University, West-Lafayette, IN, USA, \texttt{\{sun578,gscutari\}@purdue.edu.}}}

\maketitle

\begin{abstract}
  In this paper, we study \emph{distributed big-data nonconvex} optimization in
  multi-agent networks.  We consider the (constrained) minimization of the sum
  of a smooth (possibly) nonconvex function, i.e., the agents' sum-utility, plus
  a convex (possibly) nonsmooth regularizer. Our interest is in big-data
  problems wherein there is a large number of variables to optimize. If treated
  by means of standard distributed optimization algorithms, these large-scale
  problems may be intractable, due to the prohibitive local computation and
  communication burden at each node.  We propose a novel distributed solution
  method whereby at each iteration agents optimize and then communicate  (in an uncoordinated fashion)
  only a subset of their decision variables. To deal with   non-convexity
  of the cost function, the novel scheme hinges on Successive Convex
  Approximation (SCA) techniques coupled with i) a tracking mechanism
  instrumental to locally estimate gradient averages; and ii) a novel
  \emph{block-wise} consensus-based protocol to perform local block-averaging
  operations and gradient tacking. Asymptotic convergence to stationary
  solutions of the nonconvex problem is established. Finally, numerical results
  show the effectiveness of the proposed algorithm and highlight how the block
  dimension impacts on the communication overhead and practical convergence
  speed. 
\end{abstract}

\section{Introduction}
\label{sec:intro}
In many modern control, estimation, and learning applications, optimization
problems with a very-high dimensional decision variable arise. These problems
are often referred to as \emph{big-data} and call for the development of new
algorithms.
A lot of attention has been devoted in recent years to devise 
parallel, possibly
asynchronous, algorithms to solve problems of this sort on   shared-memory systems.
Distributed optimization has also  received significant
attention, since it provides   methods  to
solve cooperatively networked optimization problems   by exchanging messages only with
neighbors and without resorting to any shared memory or centralized coordination.
In this paper, we consider  \emph{distributed big-data optimization}, that is, big-data problems that must be solved by a network of agents in
a distributed way. We are not aware of any work in the literature that can address the challenges of big-data optimization over networks. 
We organize the relevant literature  in two groups, namely: (i)
centralized and parallel methods for big-data optimization, and (ii) primal distributed methods  for  multi-agent optimization.

A coordinate-descent method for huge-scale
smooth optimization problems has been introduced in~\cite{nesterov2012efficiency}, and then   extended in~\cite{richtarik2012parallel,richtarik2014iteration,necoara2016parallel} to deal
with (convex) composite objective functions  in a parallel fashion.
A parallel algorithm based on   Successive Convex Approximation (SCA) has been
proposed in~\cite{facchinei2015parallel} to cope with non-convex objective functions. 
An asynchronous extension   has been proposed 
in~\cite{cannelli2016asynchronous}.
In~\cite{mokhtari2016doubly} a parallel stochastic-gradient algorithm has been
studied wherein  each processor randomly chooses a component of the sum-utility function and updates only a  block (chosen uniformly at random)  of the entire decision variable.
Finally, an asynchronous parallel mini-batch algorithm for stochastic
optimization has been proposed in~\cite{feyzmahdavian2016asynchronous}. These algorithms, however, are not   implementable efficiently in a  distributed network setting, because either  they  assume that all agents know the whole sum-utility or that, at each iteration, each agent has access to  the other agents' variables.

In the last years,   distributed multi-agent optimization  has received significant attention. Here, we refer only to distributed primal methods, which are more closely
related to the approach proposed in this paper. 
In~\cite{necoara2013random}, a coordinate descent method to solve linearly constrained
problems over undirected networks has been proposed.
In~\cite{nedic2015distributed} a broadcast-based algorithm over time-varying digraphs has been studied, and its extension  to distributed online optimization 
has been considered in~\cite{akbari2015distributed}.
In~\cite{margellos2016distributed} a distributed algorithm for convex optimization 
over time-varying networks in the presence of constraints and uncertainty using a proximal 
minimization approach has been proposed.
A distributed algorithm, termed NEXT, combining SCA techniques with a novel gradient tracking 
mechanism instrumental to  estimate locally the average of the agents' gradients, has been proposed   
in~\cite{dilorenzo2015distributed,dilorenzo2016next} to solve nonconvex constrained optimization 
problems over time-varying networks. The scheme has been extended in ~\cite{sun2016distributed,sun2017distributed} to deal with 
directed (time-varying) graphs. 
More recently in~\cite{nedic2016geometrically} and \cite{qu2016harnessing}, special
instances of the two aforementioned algorithms have been shown to enjoy
geometric convergence rate when applied to unconstrained smooth (strongly)
convex problems.
A Newton-Raphson consensus strategy has been introduced in~\cite{zanella2012asynchronous}
to solve unconstrained, convex optimization problems leveraging asynchronous, 
symmetric gossip communications.
The same technique has been extended to design an algorithm for directed, asynchronous, and 
lossy communications networks in~\cite{carli2015analysis}.
None of these schemes can efficiently  deal with  big-data
problems. In fact, when it comes to big-data problems, local computation and
communication requirements need to be explicitly taken into account in the
algorithmic design. Specifically,  (i)
local optimization steps on the \emph{entire} decision variable cannot be performed,
because they would be too computationally demanding; and (ii) communicating,  at each iteration,  the entire vector of variables of the agents would   incur in an unaffordable communication overhead. %
%
%

First attempts to distributed big-data optimization are~\cite{carli2013distributed,notarnicola2016randomized}
for a structured, \emph{partitioned}, optimization problem, and~\cite{wang2016coordinate}
by means of a partial stochastic gradient for strongly convex smooth problems. These schemes however are not applicable to multi-agent problems wherein the agents' functions are not (partially) separable.

%
%
In this paper, we propose the first distributed algorithm for (possibly nonconvex) big-data optimization problems over  networks, modeled as  digraphs.  %
To cope with the issues posed by big-data problems, in the proposed scheme,   agents
maintain  a local estimate of the common decision variables  but, at  every iteration, they update and
communicate  only \emph{one block}. Blocks are selected  in an uncoordinated
fashion among the agents by means of an ``essentially cyclic rule'' guaranteeing that all blocks
are persistently updated.  %
Specifically,   each agent minimizes a strongly-convex local approximation of the nonconvex sum-utility  function with respect to the
selected block variable only. Inspired by the SONATA
algorithm~\cite{sun2017distributed}, the surrogate is based on the agent local
cost-function and a local  gradient estimate (of the smooth global-cost portion).
The optimization step is combined with a \emph{block-wise} consensus step, 
 tracking the network average gradient and guaranteeing the asymptotic agreement of the local solution estimates to a
common stationary solution of the nonconvex problem.

The paper is organized as follows. In Section~\ref{sec:setup_preliminaries}, we
present the problem set-up and recall some preliminaries. In
Section~\ref{sec:algorithm}, we first introduce the new block-wise consensus  protocol, 
and then formally present our novel distributed big-data
optimization algorithm along with its convergence properties. Finally, in Section~\ref{sec:simulations}, we show a numerical
example to test our algorithm.


\section{Distributed Big-Data Optimization:\\Set-up and Preliminaries}
\label{sec:setup_preliminaries}
In this section, we introduce the big-data optimization set-up and recall two
distributed optimization algorithms inspiring the one we propose in this paper.

\subsection{Distributed Big-Data Optimization Set-up}
\label{subsec:setup}
We consider a multi-agent system composed of $N$ agents, aiming at solving cooperatively the following  (possibly) nonconvex, nonsmooth,  \emph{large-scale} optimization problem\vspace{-0.1cm}
\begin{align}
\tag{P}
\begin{split}
  \min_{\bx} \,\,\,\,\,\: &\: U ( \bx) \triangleq \sum_{i=1}^N f_i ( \bx ) + \sum_{\ell = 1}^B g_\ell ( \bx_\ell )
  \\
  \subj \: &\: \bx_\ell \in \KK_\ell, \quad \ell \in\until{B},
\end{split}
\label{eq:problem}
\end{align}
where   $\bx$ is the vector  of optimization variables, partitioned in $B$ blocks\vspace{-0.2cm}
\begin{align*}
  \bx \triangleq 
  \begin{bmatrix}
    \bx_1 \\[-0.4em] \vdots \\[-0.4em] \bx_B
  \end{bmatrix},
\end{align*}
with each   $\bx_\ell\in\real^{d}$, $\ell\in\until{B}$; $\map{f_i}{\real^{dB}}{\real}$ is the cost
function of agent $i$, assumed to be smooth but (possibly) nonconvex; 
  $\map{g_\ell}{\real^{d}}{\real}$, $\ell\in\until{B}$, is 
a convex (possibly nonsmooth) function; and $\cal K_\ell$, $\ell\in\until{B}$, is a closed convex set. 
Usually the nonsmooth term in \eqref{eq:problem} is used to promote
some extra structure in the solution, such as (group) sparsity. In the following, we will 
denote by $\KK\triangleq \KK_1\times \cdots \times \KK_{B}$ the feasible set of \eqref{eq:problem}.  
We study problem~\eqref{eq:problem} under the following assumptions. 
%
%

  \begin{assumption}[On the optimization problem]\mbox{}
\begin{enumerate} 
  \item Each $\KK_\ell \neq \emptyset$  is closed and convex;
  \item Each $f_i:\mathbb{R}^{dB}\rightarrow \mathbb{R}$  is $C^1$ on (an open set containing)  $\KK$;
  \item Each $\nabla f_i$ is $L_i$-Lipschitz  continuous 
    and bounded on $\KK$;
  \item Each $g_\ell:\mathbb{R}^d\rightarrow \mathbb{R}$ is convex (possibly nonsmooth)  on (an open set containing) $\KK$,
  with bounded subgradients over $\KK$;
  \item $U$ is coercive on $\KK$, i.e., $\lim_{\bx\in\KK,\|\bx\|\to\infty} U(\bx) = \infty$. \oprocend
\end{enumerate}
\label{ass:cost_functions}
\end{assumption}
The above assumptions are quite standard and satisfied by many practical  problems; see, e.g. \cite{facchinei2015parallel}. Here, we only remark that we do not assume any convexity of $f_i$. In the following, we also make the blanket   assumption  that each agent $i$ knows only its own cost function $f_i$ (the regularizers $g_\ell$ and the feasible set $\cal K$) but not the other agents' functions.  

\smallskip
\noindent \emph{On the communication network:} 
The  communication network of the agents is  modeled as a  fixed, directed
graph $\GG = (\until{N},\EE)$, where $\EE\subseteq \until{N} \times \until{N}$
is the set of edges. The edge $(i,j)$ models the fact that agent $i$
can send a message to agent $j$.
We denote by $\nbrs_i$ the set of \emph{in-neighbors} of node $i$ in the fixed
graph $\GG$  {including itself}, i.e.,
$\nbrs_i \triangleq \left\{j \in \until{N} \mid (j,i) \in \EE \right\}{\cup \{i\}}$.  
We make the following assumption on the graph connectivity. 
\smallskip
\begin{assumption}
The graph $\GG$ is strongly connected.~\oprocend
\label{ass:strong_conn}
\end{assumption} 
\smallskip

\noindent\emph{Algorithmic Desiderata:} Our goal is to solve problem~\eqref{eq:problem} in 
a distributed fashion, leveraging local communications among neighboring agents. 
As a    major departure from current  literature on distributed optimization,  
here we focus on \emph{big-data} instances of problem~\eqref{eq:problem} wherein the vector of 
variables $\bx$ is composed of a huge number of components ($B$ is very large).  In such problems, 
minimizing the sum-utility with respect to the whole $\bx$, or  even computing the gradient or evaluating 
the value of a single  function $f_i$, can require substantial computational efforts. 
Moreover,
exchanging  an estimate of the \emph{entire} local decision variable $\bx$ over the network 
(like current distributed schemes do) is not efficient or even feasible, due to the excessive 
communication overhead.  We design next the first scheme able to deal with such challenges.
To this end, we first review two existing distributed algorithms
for nonconvex optimization that will act as building blocks for our novel algorithm. 

\subsection{NEXT and SONATA: A Quick Overview}
\label{subsec:preliminaries}

In~\cite{dilorenzo2015distributed,dilorenzo2016next} a distributed algorithm, termed NEXT, is proposed to solve
nonconvex  optimization problems in the form~\eqref{eq:problem}.
%
%
%
The algorithm is based on an iterative two-step procedure whereby all the agents first update their  local estimate  $\bx_{(i)}$  of the optimization variable
$\bx$ by solving a suitably  chosen convexification of \eqref{eq:problem}, and then communicate with their neighbors to asymptotically force an agreement among the local variables while converging to  a stationary solution of \eqref{eq:problem}.
%
The strongly convex approximation of the
original noncovex function has the following form: The  nonconvex cost
function $f_i$ is replaced by  a suitable  strongly convex  \emph{surrogate function}
$\tilde{f}_i$  (see~\cite{facchinei2015parallel}) whereas the sum of the unknown   functions of the other agents  $\sum_{j\neq i} f_j$ is approximated by a linear term whose coefficients track the gradient of $\sum_{j\neq i} f_j$.  
%
More formally, given the current iterate $\bx_{(i)}^t$,   the optimization step reads:  
\begin{align*}
  \widetilde{\bx}_{(i)}^{t} & = \argmin_{\bx \in \KK} \tilde{f}_i (\bx; \bx_{(i)}^t ) +
  (\bx - \bx_{(i)}^t )^\top \widetilde{\bpi}_{(i)}^t + g(\bx),
  \\
  \bv_{(i)}^{t} & = \bx_{(i)}^{t} +\gamma^t ( \widetilde{\bx}_{(i)}^{t} - \bx_{(i)}^{t}),
\end{align*}
where $\widetilde{\bpi}_{(i)}^t$ is a local estimate of   $\sum_{j\neq i} \nabla f_j(\bx_{(i)}^t)$ (that needs to be properly updated);  $g(\bx)\triangleq \sum_{\ell = 1}^B g_\ell ( \bx_\ell )$; and  $\gamma^t$ is the  step-size. Given $\widetilde{\bx}_{(i)}^{t}$ and $\bx_{(i)}^{t}$, the local variable is updated from $\bx_{(i)}^{t}$ along the direction 
$\widetilde{\bx}_{(i)}^{t} - \bx_{(i)}^{t}$ using the step-size value $\gamma^t$. The resulting $\bv_{(i)}^{t}$ will
then be averaged with the neighboring counterparts to enforce a consensus.

The second step of NEXT, consists in communicating with the neighbors in order to update the local decision variables as well as the gradient estimates  $\widetilde{\bpi}_{(i)}^t$. 
Formally, we have the following two consensus-based updates:
\begin{align*}
  \bx_{(i)}^{t+1} & = \sum_{j\in \mathcal N_i}   a_{ij}  \bv_{(j)}^{t},  
  \\
  \by_{(i)}^{t+1} & = \sum_{j\in \mathcal N_i} a_{ij}  \by_{(j)}^{t} 
  +\nabla f_i (\bx_{(i)}^{t+1}) - \nabla f_i (\bx_{(i)}^{t}),
  \\
  \widetilde{\bpi}_{(i)}^{t+1} & = N \cdot \by_{(i)}^t \!-\! \nabla f_i (\bx_{(i)}^{t+1}),
\end{align*}
where $\by_{(i)}^{t}$ is an auxiliary local variable (exchanged among neighbors) instrumental to update $\widetilde{\bpi}_{(i)}^t$; and  $\bA \triangleq [a_{ij}]_{i,j=1}^N$ is a \emph{doubly-stochastic}
matrix that matches the   communication graph $\mathcal G$: $a_{ij}>0$,   if $j\in \mathcal N_i$; and $a_{ij}=0$ otherwise.

Note that NEXT requires the weight matrices $\bA$ to be doubly-stochastic, which limits the applicability of the algorithm to  directed graphs.  
SONATA, proposed in \cite{sun2016distributed,sun2017distributed}, overcomes this limitation by replacing the plain consensus scheme of NEXT with a push-sum-like mechanism, which recovers dynamic
average consensus of the local decision variable $\bx_{(i)}^{t}$ and 
gradient estimates $\by_{(i)}^{t}$ by using only column stochastic weight matrices. Introducing an auxiliary local variable
$\phi_{(i)}^t$ at each agent's side, the consensus protocol of SONATA reads
\begin{align*}
  \phi_{(i)}^{t+1} & = \sum_{j\in \mathcal N_i} a_{ij}  \phi_{(j)}^{t},
  \\
  \bx_{(i)}^{t+1} & = \frac{1}{\phi_{(i)}^{t+1}} \sum_{j\in \mathcal N_i}  a_{ij}  \phi_{(j)}^{t} \bv_{(j)}^{t},
  \\
  \by_{(i)}^{t+1} & = \frac{1}{\phi_{(i)}^{t+1}} \! \left( \sum_{j\in \mathcal N_i} \! a_{ij}\, \phi_{(j)}^{t} \by_{(j)}^{t}
  + \nabla f_i (\bx_{(i)}^{t+1}) - \nabla f_i (\bx_{(i)}^{t}) \! \right)\!,
  \\
  \widetilde{\bpi}_{(i)}^{t+1} & = N \cdot \by_{(i)}^t - \nabla f_i (\bx_{(i)}^{t+1}),
\end{align*}
where now $ \bA \triangleq [a_{ij}]_{i,j=1}^N$ is only \emph{column
stochastic} (still matching the graph $\mathcal G$).




\section{\algfullname/ Distributed Algorithm}
\label{sec:algorithm}
In this section we introduce our distributed big-data optimization algorithm.
 Differently from current distributed methods, our algorithm performs  \emph{block-wise} updates and communications.  A building block of the proposed scheme is a   block-wise consensus protocol of independent interest, which is introduced next (cf. Sec. \ref{subsec:block_consensus}). Then, we will be ready to describe our new algorithm (cf. Sec.~\ref{Sec:B_SONATA}).
 
%

\subsection{Block-wise Consensus}
\label{subsec:block_consensus}
We propose a push-sum-like scheme that acts at the level of each block $\ell\in\until{B}$.
%
Specifically, consider a system of $N$ agents, whose communication network is modeled as a digraph $\cal G$ satisfying Assumption \ref{ass:strong_conn}; and let the agents aim at agreeing on the (weighted) average value of their
initial states $\bx_{(i)}^0$, $i\in\until{N}$. While  at each
iteration $t$  agents can update their entire vector $\bx_{(i)}^t$, they can however send to
their neighbors \emph{only one block}; let denote by $\bx_{(i,\ell_i^t)}$ the block $\ell_i^t$ that, at time $t$, agent $i$ selects (according to a suitably chosen rule) and sends to its neighbors,   with $\ell_i^t \in \until{B}$.
Thus, at each iteration, agent $i$ runs a consensus protocol on the $\ell$-th
block by using only information received from in-neighbors that have sent
block $\ell$ at time $t$ (if any).
A natural way to model this protocol is to introduce a \emph{block-dependent}
neighbor set, defined as
\begin{align*}
  \nbrs_{i,\ell}^t \triangleq \{ j \in \nbrs_i \mid \ell_j^t = \ell \}  \cup \{i \} \subseteq \nbrs_i,
\end{align*}
which includes, besides agent $i$, only the in-neighbors of agent $i$ in $\GG$ that have
sent block $\ell$ at time $t$.
Consistently, we denote by $\GG_\ell^t\triangleq (\until{N},\EE_\ell^t)$ the \emph{time-varying} subgraph of
$\GG$ associated to block $\ell$ at iteration $t$. Its edge set is
\begin{align*}
  \EE_\ell^t \triangleq \{(j,i)\in\EE \mid j\in \nbrs_{i,\ell}^t, i\in\until{N}\}.
\end{align*}
Following the idea of consensus protocols over time-varying digraphs, we
introduce a weight matrix $\bA_\ell^t \triangleq [a_{ij\ell}^t]_{i,j=1}^N$ matching $\GG_\ell^t$, such
that $a_{ij\ell}^t \in [\theta, 1]$, for some $\theta \in (0,1)$, 
if $j \in \nbrs_{i,\ell}^t$; and $a_{ij\ell}^t = 0$
otherwise.  Using $\bA_\ell^t$, we can rewrite the consensus scheme of SONATA block-wise as

\begin{align}
\label{eq:block_consensus}
\begin{split}
  \phi_{(i,\ell)}^{t+1}  & = 
  \sum_{j\in \innbrs_{i} } a_{ij\ell}^t \phi_{(j,\ell)}^t, \hspace{1.9cm} \ell \in \until{B},
  \\
  \bx_{(i,\ell)}^{t+1} & = 
  \dfrac{1}{\phi_{(i,\ell)}^{t+1}} 
  \sum_{j\in \innbrs_{i}}
   a_{ij\ell}^t \phi_{(j,\ell)}^t\,\bx_{(j,\ell)}^t, \hspace{0.2cm} \ell\in \{1,\ldots, B\},
\end{split}
\end{align}
%
where $\phi_{(i,\ell)}^0$ and $\bx_{(i,\ell)}^0$ are given, for all $i\in\until{N}$.

We study now under which conditions the block consensus protocol (\ref{eq:block_consensus}) reaches an asymptotic agreement. 

 For push-sum-like algorithms (as the one just stated) to achieve asymptotic
  consensus, the following standard assumption is needed.
\smallskip 
\begin{assumption}
  For all $\ell\in\until{B}$ and $t\geq 0$, the matrix $\bA_\ell^t$ is column stochastic, that is, $\mathbf{1}^\top \bA_\ell^t = \mathbf{1}^\top$.\oprocend
  \label{ass:col_stoch}
\end{assumption}
\smallskip 


We show next  how nodes can \emph{locally} build a matrix $\bA_\ell^t$ satisfying 
Assumption~\ref{ass:col_stoch} for each time-varying, directed graph $\GG_\ell^t$.
Since in our distributed optimization algorithm we work with a static,
strongly connected   digraph $\GG$ (cf.\,Assumption~\ref{ass:strong_conn}), we assume
that a column stochastic matrix $\tilde{\bA}$ that matches $\GG$ is available,
i.e., $\tilde{a}_{ij} > 0$ if $(j,i)\in\EE$ and $\tilde{a}_{ij} = 0$
otherwise, and $\mathbf{1}^\top \tilde{\bA} = \mathbf{1}^\top\tilde{\bA}$.

To show how $\bA_\ell^t$ can be constructed in a distributed way,
we start observing that at iteration $t$, an agent $j$ either sends a block $\ell$ to all
its out-neighbors in $\GG$, $\ell=\ell_j^t$, or to none, $\ell\neq\ell_j^t$.
Thus, let us concentrate on the $j$-th column $\bA_\ell^t(:,j)$ of  $\bA_\ell^t$. If
agent $j$ does not send
block $\ell$ at iteration $t$, $\ell\neq\ell_j^t$, then all elements of
$\bA_\ell^t(:,j)$ will be zero except $a_{jj\ell}^t$. Thus, to be the $j$-th column stochastic, it must be  $\mathbf{1}^\top  \bA_\ell^t(:,j) = a_{jj\ell}^t=1$ (i.e.,
$\bA_\ell^t(:,j)$ is the $j$-th vector of the canonical basis). 
Viceversa, if $j$ sends block $\ell$, all its out-neighbors in $\GG$ will
receive it and, thus, column $\bA_\ell^t(:,j)$ has the same nonzero entries as
column $\tilde{\bA}(:,j)$ of $\tilde{\bA}$. Since $\tilde{\bA}$ is column
stochastic, the same entries can be chosen, that is,
$\bA_\ell^t(:,j)=\tilde{\bA}(:j)$.
%
%
%
This rule can be stated from the point of view of each agent $i$ and its
in-neighbors, thus showing that each agent can locally construct its own
weights. 
For each $i\in\until{N}$ and $\ell\in\until{B}$, 
weights $a_{ij\ell}^t$ can be defined as
\begin{eqnarray}
  a_{ij\ell}^t \triangleq 
  \begin{cases}
	  \tilde{a}_{ij}, & \text{if } j \in \nbrs_i \text{ and } \ell = \ell_j^t,
	  \\
	  1,                    & \text{if } j=i \text{ and } \ell \neq \ell_j^t, 
	  \\
	  0,                    & \textnormal{otherwise}.
	\end{cases}
	\label{eq:weights}
\end{eqnarray}
%


Besides imposing $\bA_\ell^t$ to be column stochastic for each $t$, another key
aspect to achieve   consensus is that the time-varying digraphs
$\GG_\ell^t$ be $T$-strongly connected, i.e., for all $t\ge 0$ the union digraph
$\bigcup_{s=0}^{ T-1} \GG_\ell^{t+s}$ is strongly connected.

Since each (time-varying) digraph $\GG_\ell^t$ is induced by the block selection rule, its connectivity 
properties are determined by the block selection rule. Thus,   the $T$-strongly connectivity requirement imposes a condition on the way the blocks are selected. 
The following general  \emph{essentially cyclic}  rule is enough to meet this requirement. 
%
%
\begin{assumption}[Block Updating Rule]
	For each 	agent $i\in\until{N}$, there exists a (finite) constant $T_i>0$ such that \vspace{-0.3cm}
	\begin{equation*} 
	  \bigcup_{s=0}^{T_i-1} \{\ell_i^{t+s} \} = \until{B}, \text{ for all } t \ge 0.   \eqoprocend\vspace{-0.3cm}
  \end{equation*}
\label{ass:block_selection}
\end{assumption}
\smallskip 
Note that the above rule does not impose any coordination among the agents, but agents selects their own block independently. Therefore, at the same iteration, different agents may update different blocks. Moreover, some blocks can be updated more often than others. 
However, the rule  guarantees that, within a finite time window of length $T \le \max_{i\in\until{N}} T_i$
all the blocks have been updated at least once by all agents. This is enough for $\GG_\ell^t$ to be $T$-strongly connected, as stated next. 

%

\begin{proposition}
Under Assumption~\ref{ass:strong_conn} and \ref{ass:block_selection}, there exists a  $0<T\leq \max\limits_{i\in\until{N}} \!T_i$, such that $\bigcup_{s = 0}^{ T-1} \GG_\ell^{t+s}$, $\ell\in \{1,\ldots, B\}$, is   strongly connected, for all $t\geq 0$. 
\end{proposition}
 
\begin{proof}
Consider a particular block $\ell$, and define $s_i^t(\ell)$ as the last time agent $i$ sends block $\ell$ in the  time window $[t,t+T-1]$, where $t \geq 0$. The essentially cyclic rule (cf.\, Assumption~\ref{ass:block_selection}) implies that $s_i^t(\ell) \leq T-1 $ for all $i\in \until{N}$. By   definition of $\GG_\ell^t$, we have that any edge $(j,i) \in \EE$ also belongs to $G_\ell^{t+ s_i^t(\ell)}$. Since $\EE \subseteq \bigcup_{i\in\VV} \GG_\ell^{t+ s_i^t(\ell)}\subseteq \bigcup_{\tau=t}^{t+T-1} \GG_\ell^\tau$, we have $\bigcup_{s = 0}^{ T-1} \GG_\ell^{t+s}$ is strongly connected, since $\GG$ is so   (cf.\,Assumption~\ref{ass:strong_conn}).
\end{proof}\smallskip

Since $\{\GG_{\ell}^t\}_{t\in\natural}$ is $T$-strongly connected for all $\ell \in \until{B}$ and each $\bA_\ell^t$ is a column stochastic matrix matching    $\GG_{\ell}^t$, a direct application of \cite[Corollary 2]{nedic2015distributed} leads to the following  convergence results for the matrix product $\bA_\ell^{t+T:t} \triangleq \bA_\ell^{t+T} \bA_\ell^{t+T-1} \cdots \bA_\ell^t$.
  
\begin{proposition}
Suppose that  Assumptions~\ref{ass:strong_conn} and \ref{ass:block_selection} hold true, and let  $\bA_\ell^t$ be defined as in \eqref{eq:weights}, with $\ell\in \{1,\dots, B\}$. Then,  the matrix product $\bA_\ell^{t+T:t} \triangleq \bA_\ell^{t+T} \bA_\ell^{t+T-1} \cdots \bA_\ell^t$
  satisfies the geometrical decay property, i.e.,
  there exists a sequence of (real) stochastic vectors $   \{\bxi^t_\ell\}_{t\in\natural}$ such that   
  \begin{align*}
    \big | (\bA_\ell^{t+T:t})_{ij}   - (\bxi_\ell^t)_i \big | 
    \le c_\ell (\rho_\ell)^T,
  \end{align*}
  for some $c_\ell>0$ and $\rho_\ell\in (0,1)$. 
   \label{prop:induced_graphs}
\end{proposition}

 Invoking  Proposition \ref{prop:induced_graphs},  we can finally obtain the desired convergence results for the proposed block consensus protocol (\ref{eq:block_consensus}). 
\begin{theorem} Consider the block consensus scheme (\ref{eq:block_consensus}), 
under Assumptions~\ref{ass:strong_conn} and \ref{ass:block_selection}, and let each  weight matrix $\bA_\ell^t$ be defined according to \eqref{eq:weights}. Then, the following holds: 
\begin{equation*}
  \lim_{t\to\infty} \Big\| 
    \bx_{(i,\ell)}^t  -\textstyle \sum\limits_{i=1}^N \phi_{(i,\ell)}^0 \bx_{(i,\ell)}^0 
  \Big \|=0,\: \: \forall\, \ell \in \until{B}. \eqoprocend
\end{equation*}\smallskip 
\end{theorem}

\subsection{Algorithm design: \algfullname/}\label{Sec:B_SONATA}
We are now in the position to introduce our new big-data algorithm, which combines 
SONATA (suitably tailored to a block implementation) and the proposed block consensus scheme. 
Specifically,
%
each agent maintains a set of local and auxiliary variables, which are exactly the
same as in SONATA (cf.\,Section~\ref{subsec:preliminaries}),
namely $\bx_{(i)}^t$, $\bv_{(i)}^t$, $\bphi_{(i)}^t$, $\by_{(i)}^t$ and
$\tbpi_{(i)}^t$. Consistently with the block structure of the optimization
variable, we partition accordingly these variables. 
%
We denote by $\bx_{(i,\ell)}^t\in\real^{d}$ the $\ell$-th
block-component of local estimate $\bx_{(i)}^t$ that agent $i$ has at time $t$; we use the same notation for the blocks of the other vectors. 

%

Informally, agent $i$ first  performs a   minimization only with respect to the
block-variable it selects; 
%
%
then, it  performs the block-wise consensus update introduced in the
previous subsection. 
The \algacronym/ distributed algorithm is formally reported in the table below
(from the perspective of node $i$) and discussed in details afterwards.
Each $i$ initializes the local states as: $\bx_{(i)}^0$ to an arbitrary value,
$\bphi_{(i)}^0 \!=\! [1,\!\ldots\!,1]^\top \triangleq \1$, and
$\by_{(i)}^0 \!=\! \nabla f_i(\bx_{(i)}^0)$.

%
\begin{algorithm}[!ht]
\renewcommand{\thealgorithm}{}
\floatname{algorithm}{Distributed Algorithm}
  \begin{algorithmic}
      
    \StatexIndent[0] \textbf{Optimization}:

    \StatexIndent[0.25] Select $\ell_i^t \in \until{B}$ and compute
      \begin{align}
        \label{eq:alg_x_min}
        \begin{split}
        & \widetilde{\bx}_{(i,\ell_i^t)}^t \! \! \triangleq \!
          \argmin_{ \bx_{\ell_i^t} \in \KK_{\ell_i^t} } \!
          \tf_{i,\ell_i^t} \big(\bx_{\ell_i^t}; \bx_{(i)}^t \big)
          \!+\! ( \bx_{\ell_i^t} \!-\! \bx_{(i,\ell_i^t)}^t )^{\! \top} \! \widetilde{\bpi}_{(i,\ell_i^t)}^t
          \!+\!
          g_{\ell_i^t} ( \bx_{\ell_i^t} ),
        \end{split}
        \\
        \label{eq:alg_v}
        & \bv_{(i,\ell_i^t)}^{t} = \bx_{(i,\ell_i^t)}^t + \gamma^t ( \widetilde{\bx}_{(i,\ell_i^t)}^t - \bx_{(i,\ell_i^t)}^t );
      \end{align}

    \medskip
    
    \StatexIndent[0] \textbf{Communication}: 
     \StatexIndent[0.25] Broadcast $\bv_{(i,\ell_i^t)}^t$,
      $\phi_{(j,\ell_i^t)}^t$, $\by_{(j,\ell_i^t)}^t$ to the
      out-neighbors 
    \StatexIndent[0.25] \textsc{for} \!$\ell \!\in \!\until{B}$:\!
    receive $\phi_{(j,\ell)}^t$,$\bv_{(j,\ell)}^t$,%
    $\by_{(j,\ell)}^t$ from $j \!\in\! \nbrs_{i,\ell}^t$, and update:
    \begin{align}
      \phi_{(i,\ell)}^{t+1} & = \sum_{j \in \innbrs_i } a_{ij\ell}^t \,\phi_{(j,\ell)}^t,
    \label{eq:alg_phi_update}
    \\
      \bx_{(i,\ell)}^{t+1} & = 
      \dfrac{1}{\phi_{(i,\ell)}^{t+1}} \sum_{j\in \innbrs_i } a_{ij\ell}^t\, \phi_{(j,\ell)}^t \bv_{(j,\ell)}^t,
    \label{eq:alg_x_update}
    \\
    \begin{split}
      \by_{(i,\ell)}^{t+1} & = 
      \dfrac{1}{\phi_{(i,\ell)}^{t+1}} 
      \Big(\sum_{j \in \innbrs_i}
      \!\!
       a_{ij\ell}^t\, \phi_{(j,\ell)}^t \by_{(j,\ell)}^t 
      +  \nabla_\ell f_i \big( \bx_{(i)}^{t+1} \big) 
       -
       \nabla_\ell f_i \big( \bx_{(i)}^{t} \big) \Big)
    \end{split}
    \label{eq:alg_y_update}
      \\
      \widetilde{\bpi }_{(i,\ell)}^{t+1} &= N \cdot \by_{(i,\ell)}^{t+1} - \nabla_\ell f_i \big( \bx_{(i)}^{t+1} \big).
    \label{eq:alg_pi_update}
    \end{align}

  \end{algorithmic}
  \caption{\algfullname/}
  \label{alg:algorithm}
\end{algorithm}

We  discuss now the steps of the algorithm. 
At  iteration $t$, each agent $i$ selects a block $\ell_i^t \in \until{B}$
according to a rule satisfying Assumption~\ref{ass:block_selection}.
Then, a local approximation of problem~\eqref{eq:problem} is constructed by (i)
replacing the nonconvex cost $f_i$ with a strongly convex surrogate
$\tilde{f}_{(i,\ell_i^t)}$ depending on the current iterate $\bx_{(i)}^t$,
and, (ii) adding a gradient estimate $\widetilde{\bpi}_{(i,\ell_i^t)}^t$ of the
remainder cost functions $f_j$, with $j\neq i$. How the surrogate can be
constructed will be clarified in the next subsection.
Agent $i$ then solves problem~\eqref{eq:alg_x_min} with respect to its own  block  $\bx_{\ell_i^t}$  only, and then 
using the solution $\widetilde{\bx}_{(i,\ell_i^t) }^t$ it updates only the
$\ell_i^t$-th block of the auxiliary state $\bv_{(i,\ell_i^t)}^{t}$, according to \eqref{eq:alg_v}, where 
  $\gamma^t$ is a suitably chosen  (diminishing) step-size.
After this update, each node $i$ broadcasts  to its out-neighbors only
$\bv_{(i,\ell_i^t)}^{t}$. We remark that this, together with $\by_{(i,\ell_i^t)}^{t}$ and $\phi_{(i,\ell_i^t)}^t$, are  the sole quantities sent by
agent $i$ to its neighbors.

As for the consensus step, agent $i$ updates block-wise its local variables by
means of the novel block-wise consensus protocol described in
Section~\ref{subsec:block_consensus}. 

We conclude this discussion by pointing out that in order to perform
update~\eqref{eq:alg_y_update}, agent $i$ needs to update all blocks of
$\bx_{(i)}^{t+1}$. And the same holds for~\eqref{eq:alg_pi_update} and
$\by_{(i)}^{t+1}$. However, these updates involve only the blocks $\ell_j^t$
with $j\in\nbrs_i$.

\subsection{Convergence of \algacronym/}
We provide now the main convergence result of \algacronym/.
Convergence is guaranteed under mild (quite standard) assumptions on the surrogate 
functions $\tf_{i,\ell}$ [cf. (\ref{eq:alg_x_min})] and the step-size sequence 
$\{\gamma^t\}$ [cf. (\ref{eq:alg_v})]. Specifically, we need the following. 
%
\begin{assumption}[On the surrogate functions]
Given problem~\eqref{eq:problem} under Assumption~\ref{ass:cost_functions}, 
each surrogate function 
$\tf_{i,\ell}: \KK_\ell\times \KK\rightarrow \real$ is chosen so that 
\begin{enumerate}
\item 
  $\tf_{i,\ell} (\bullet;\bx)$ is uniformly strongly convex with constant $\tau_i >0$
  on $\KK_\ell$;
\item
  $\nabla \tf_{i,\ell} (\bx_{\ell};\bx) = \nabla_\ell f_i (\bx)$, for all $\bx \in \KK$;
\item
  $\nabla \tf_{i,\ell} (\bx_{\ell}; \bullet)$ is uniformly Lipschitz continuous on $\KK$;
\end{enumerate}
where $\nabla \tf_{i,\ell}$ denotes the partial gradient of $\tf_{i,\ell}$ 
with respect to its first argument.~\oprocend
\label{ass:surrogate}
\end{assumption}

Conditions (i)-(iii) above are mild assumptions: $\tf_{i,\ell} (\bullet;\bx)$ should be regarded as
a (simple) convex, local, approximation of $f_i(\bullet,\bx_{-\ell})$  that preserves the first order 
properties of $f_i$ at the point $\bx$, where 
$\bx_{-\ell} \triangleq (\bx_{1},\ldots, \bx_{\ell-1}, \bx_{\ell+1},\ldots ,\bx_{B})$. 
Condition (iii) is a simple Lipschitzianity requirement
that is readily satisfied if, for example, the set $\KK$ is bounded. Several valid
instances of $\tf_{i,\ell}$ are possible for a given $f_i$; the appropriate one depends
on the problem at hand and computational requirements. 
We briefly discuss some valid surrogates below and
refer the reader to~\cite{facchinei2015parallel} and \cite{dilorenzo2016next} for further examples.
Given $\bx_{(i)}^t\in \KK$, a choice that is always possible is 
$\tilde f_{i,\ell} (\bx_{\ell}; \bx_{(i)}^t) = 
\nabla_\ell f_{i}( \bx_{(i)}^t )^\top (\bx_{\ell} - \bx_{(i,\ell)}^t ) 
+ \tau_i \| \bx_{\ell} - \bx_{(i,\ell)}^t \|^2$,
where $ \tau_i$ is a
positive constant, which leads to  the classical (block) proximal-gradient update. However, one can go  
beyond the proximal-gradient choice using $\tf_{i,\ell}$ that better exploit the structure of $f_i$; 
some examples are the following:\smallskip
 
  \noindent $\bullet$ If $f_i$ is block-wise uniformly convex, instead of linearizing $f_i$
	one can employ a second-order approximation and set
	$\tf_{i,\ell}( \bx_{\ell}; \bx_{(i)}^t) = 
	f_i(\bx_{(i)}^t) + \nabla_\ell f_i ( \bx_{(i)}^t )^\top (\bx_{\ell}  -\bx^t_{(i,\ell)})+\frac{1}{2}
	( \bx_{\ell} - \bx^t_{(i,\ell)})^\top \nabla_{\ell\ell}^2\,f_i (\bx_{(i)}^t)
	(\bx_{\ell}-\bx^t_{(i,\ell)} )+\tau_i\,\|\bx_{\ell}-\bx_{(i,\ell)}^t\|^2$;
	
	\noindent$\bullet$ In the same setting as above, one can also better preserve the partial convexity of $f_i$ and set
	$\tilde{f}_{i,\ell}( \bx_{\ell};\bx_{(i)}^t) = f_i ( \bx_{\ell}, \bx_{(i,-\ell)}^t)+\tau_i \|\bx_{\ell}-\bx_{(i,\ell)}^t\|^2$;
	  	
	\noindent$\bullet$ As a last example, suppose that $f_i$ is the difference of two convex functions $f^{(a)}_i$ and
	$f_i^{(b)}$, i.e., $f_i = f^{(a)}_i - f^{(b)} $, one can preserve the partial convexity in $f_i$ by setting
	$\tf_{i,\ell} ( \bx_\ell;\bx_{(i)}^t) = f^{(a)}_i( \bx_\ell, \bx_{(i,-	\ell)}^t) 
	-
	\nabla_\ell f^{(b)}_i( \bx_{(i)}^t)^{\!\top}\! (\bx_\ell-\bx_{(i,\ell)}^t ) 
	+
	\tau_i\|	\bx_\ell -\bx_{(i,\ell)}^t \|^2\!\!.$\smallskip

%
As for the step-size $\{\gamma^t\}$, we need the following.
\begin{assumption}[On the step-size]
The sequence $\{\gamma^t\}$, with each $0< {\gamma^t} \le 1$, satisfies:
 \begin{enumerate}
\item  
 $\eta\gamma^t\le  \gamma^{t+1} \le \gamma^t$,
  for all $t\geq 0$ and  some $\eta\in (0,1)$;\smallskip 
  \item   $\sum\limits_{t=0}^{\infty}\gamma^t = \infty$ and $
  \sum\limits_{t=0}^{\infty} (\gamma^t)^2 < \infty$.
  \oprocend
  \end{enumerate} 
\label{ass:step-size}
\end{assumption}
Condition (ii) is standard and satisfied by most practical diminishing stepsize rules. 
The upper bound condition in (i) just states that the sequence is nonincreasing 
whereas the lower bound condition dictates that $\sum_{t=0}^\kappa \gamma^t \geq \gamma^0
(\eta^0+\eta^1 + \cdots + \eta^\kappa)$, that is,
the partial sums $\sum_{t=0}^\kappa \gamma^t$ must be minorized by a {\em convergent}
geometric series. This impose a maximum decay rate to $\{\gamma^t\}$. Given that
$\sum_{t=0}^\infty \gamma^t = +\infty,$ the aforementioned requirement is clearly very mild 
and indeed it is   satisfied by most classical diminishing stepsize rules. For example, the following 
rule    satisfies Assumption \ref{ass:step-size} and has been found 
very effective  in our experiments \cite{facchinei2015parallel}: $\gamma^{t+1} = \gamma^t\left(1-\mu \gamma^t \right),$ 
with $\gamma^0\in\left(0,1\right]$ and  $\mu\in\left(0,1/\gamma^0\right)$.

We are now in the position to state the main convergence result, as given below, 
where we introduced the block diagonal weight matrix 
$\boldsymbol{\phi}_{(i)}^t\triangleq  \text{diag}\{(\phi_{(i,\ell)}^t)_{\ell=1}^B\}\kron \bI_d$,
 where $\text{diag}\{(\phi_{(i,\ell)}^t)_{\ell=1}^B\}$ denotes the diagonal matrix whose 
 diagonal entries are the components   $\phi_{(i,\ell)}$, for $\ell=1,\ldots, B$, and 
 $\kron$ is the Kronecker product. 

\begin{theorem}
  Let  $\{(\bx_{(i)}^t)_{i=1}^N\}$ 
  be the sequence  generated 
  by \algacronym/, and let  $\avg^t \triangleq (1/N)\,\sum_{i=1}^N \boldsymbol{\phi}_{(i)}^t \bx_{(i)}^t$. 
  Suppose that Assumptions~\ref{ass:cost_functions}, \ref{ass:strong_conn}, \ref{ass:col_stoch}, \ref{ass:block_selection}, \ref{ass:surrogate}, and \ref{ass:step-size} are satisfied. 
Then, the following hold:
  
  \noindent (i) \texttt{convergence}: $\{\avg^t\}$ is bounded and every of its  limit points  
  is a stationary solution of problem~\eqref{eq:problem};

  \noindent (ii) \texttt{consensus}: $\| \bx_{(i)}^t - \avg^t \| \to 0$ as $t\to \infty$, for all $i\in\until{N}$;
  \oprocend
  \label{thm:convergence}
\end{theorem}

Theorem~\ref{thm:convergence} states two results. 
First, every limit point of the weighted average   $\avg^t$ belongs to the set $\cal S$ of stationary solutions of 
problem~\eqref{eq:problem}. Second, consensus is asymptotically achieved among the local 
estimates $\bx_{(i)}^t$ over all the blocks. Therefore, every limit point of the sequence $\{(\bx_{(i)}^t)_{i=1}^N\}$ converges 
to the set $\{\1_N \kron \bx^\ast\, : \bx^\ast \in \cal S\}$. In particular, if $U$ in~\eqref{eq:problem}
is convex, \algacronym/ converges (in the aforementioned sense) to the set of global optimal solutions 
of the convex problem.


\section{Application to Sparse Regression}
\label{sec:simulations}
In this section we apply \algacronym/ to the distributed sparse regression
problem. Consider a network of $N$ agents taking linear measurements of 
a sparse signal $\bx_0\in \real^m$, with data  matrix $\bD_i\in \real^{n_i \times m}$. 
The observation taken by agent $i$ can be expressed as $\bb_i = \bD_i\bx_0 + \bn_i$,
where $\bn_i\in \real^{n_i}$ accounts for the measurement noise. To estimate the underlying 
signal $\bx_0$, we formulate the problem  as:
\begin{align}
\label{regression}
  \min_{\bx \in \KK} \: & \: 
  \sum\limits_{i=1}^N \underbrace{\| \bD_i \bx - \bb_i \|^2_2}_{ f_i(\bx) }
 \:+\: \lambda \cdot g(\bx),
\end{align}
where $\bx \in \real^{m}$; $\KK$ is the box constraint set $\KK \triangleq [k_L,k_U]^m$, with $k_L\leq k_U$; 
and $g:\real^m \to \real$ is a difference-of-convex (DC) sparsity-promoting regularizer,  given by 
\begin{equation*}
  g(\bx) \triangleq \sum_{j=1}^m g_0(x_j),\quad g_0(x_j) \triangleq \frac{\log(1+\theta |x_j|)}{\log(1+\theta)}.
\end{equation*}

The first step to apply \algfullname/ is to build a valid surrogate $\tilde{f}_{i,\ell}$ of $f_i$ (cf. Assumption \ref{ass:surrogate}). To this end, we first   rewrite   $g_0$ as a DC function:
\begin{equation*}
g_0(x) = \underbrace{\eta(\theta) |x|}_{g_0^+(x)} - \underbrace{(\eta(\theta)|x| - g_0(x))}_{g_0^-(x)},
\end{equation*}
where $g_0^+ : \real \to \real $ is  convex non-smooth with 
\begin{equation*}
\eta(\theta ) \triangleq \frac{\theta}{\log(1+\theta)},
\end{equation*}
and $g_0^- : \real \to \real $ is convex and has Lipschitz continuous first order derivative given by  
\begin{equation*}
  \frac{d g_0^-}{dx} (x) = \sign(x)\cdot \frac{\theta^2 |x|}{\log(1+\theta)(1+\theta |x|)}.
\end{equation*}

Denoting the coordinates associate with block $\ell$ as $\II_\ell$,   define matrix $\bD_{i,\ell}$  [resp. $\bD_{i,-\ell}$]  constructed by picking  the columns of $\bD_i$ that belong [resp. do not belong] to $\II_\ell$. Then, the following functions are two  valid  surrogate functions.

\noindent$-$ \textit{Partial linearization:} Since $f_i$ is convex, a first natural choice to satisfy Assumption \ref{ass:surrogate} is to keep   $f_i$ unaltered while  linearizing the nonconvex part in $g_0$, which leads to the following surrogate  
\begin{equation}
\begin{aligned}
 &\tilde{f}^{PL}_{i,\ell}(\bx_{(i,\ell)};\bx_{(i)}^t) \\
    &= \| \bD_{i,\ell} \bx_{(i,\ell)} + \bD_{i,-\ell} \bx_{(i,-\ell)}^t  - \bb_i \|^2_2  \\
 & \quad+ \frac{\tau_i^{PL}}{2} \| \bx_{(i,\ell)} - \bx_{(i,\ell)}^t\|^2 - \sum_{j\in \II_\ell} \left( w_{ij}^t(x_{(i,j)} - x_{(i,j)}^t) \right),
\end{aligned}\label{eq:sparse_reg_PL}
\end{equation}
with $w_{ij}^t \triangleq \frac{d g_0^-}{dx} (x_{(i,j)} ^t)$.

 
\noindent$-$ \textit{Linearization:}  An alternative valid  surrogate can be obtained by linearizing also $f_i$,   which leads to 
\begin{equation}
\begin{aligned}
 & \tilde{f}^{L}_{i,\ell}(\bx_{(i,\ell)};\bx_{(i)}^t) \\
&=   \left(2\bD_{i,\ell}^\top (\bD_i - \bb_i)\right)^{\!\!\top}\! (\bx_{(i,\ell)} - \bx_{(i,\ell)}^t) \!+\! \frac{\tau_i^{L}}{2} \| \bx_{(i,\ell)} - \bx_{(i,\ell)}^t\|^2\\
& \quad- \sum_{j\in \II_\ell} \left( w_{ij}^t(x_{(i,j)} - x_{(i,j)}^t) \right),
\end{aligned}\label{eq:sparse_reg_L}
\end{equation}
where $w_{ij}^t$ is defined  as in \eqref{eq:sparse_reg_PL}.

Note that the minimizer of $\tilde{f}^{L}_{i,\ell}$ can be computed in closed form, and it is given by 
\begin{equation*}
  \bx_{(i,\ell)}^{t+1} = \mathcal{P}_{\KK_\ell} \left( \! \SS_{\frac{\lambda \eta}{\tau_i^L} } 
  \Big\{ \bx_{(i,\ell)}^t -\! \frac{1}{\tau_i^L}(2\bD_{i,\ell}^T (\bD_i - \bb_i) - \bw^t_{i,\ell})
  \Big\}
  \! \! \right)\!,
\end{equation*}
where $\bw^t_{i,\ell} \triangleq (w_{ij}^t)_{j\in \II_\ell}$, $\SS_\lambda(\bx) \triangleq \sign (\bx) \cdot \max\{ |\bx| - \lambda,0\}$ 
(operations are performed element-wise), and $\mathcal{P}_{\KK_\ell}$ is the Euclidean projection  onto the convex set $\KK_\ell$.

We term the two versions of  \algacronym/ based on \eqref{eq:sparse_reg_PL} and \eqref{eq:sparse_reg_L} 
\algacronym/-PL and \algacronym/-L, respectively.

We test our algorithms under the following   simulation set-up.  
The variable dimension $m$ is set to be $2000$, $\KK=[-10,10]^{2000}$, and the regularization parameters are set to  $\lambda=0.1$
and $\theta=20$. The network is composed of $N=50$ agents, communicating over a fixed undirected graph $\GG$,
generated using an Erd\H{o}s-R\'enyi random having algebraic connectivity $6$. The components of the ground-truth signal $\bx_0$ are i.i.d.   generated according to  the Normal distribution  $\NN(0,1)$. To impose sparsity on $\bx_0$, we set the smallest $80\%$ of the entries of $\bx_0$ to zero.
Each agent $i$ 
has a measurement matrix $\bD_i \in \real^{400\times 2000}$ with 
i.i.d.  $\NN(0,1)$ distributed entries   (with $\ell_2$-normalized rows), and 
 the observation  noise $\bn_i$ has entries i.i.d. distributed according to $\NN (0,0.1)$.

 {
We compare   \algacronym/-PL and \algacronym/-L   with a non-block-wise distributed (sub)-gradient-projection
algorithm, constructed by adapting the sub-gradient-push in~\cite{nedic2015distributed} to a constrained nonconvex problem 
according to the protocol   in~\cite{bianchi2013convergence}.} We term such a scheme D-Grad. 
Note that there is no formal proof of convergence of D-Grad in the nonconvex setting.   
We used the following tuning for the algorithms. 
The diminishing step-size is chosen as
\begin{align*}
  \gamma^t = \gamma^{t-1} (1 - \mu \gamma^{t-1}),
\end{align*}
with $\gamma^0 = 0.5$ and $\mu =10^{-5}$; the proximal parameter for  \algacronym/-PL and  \algacronym/-L   is chosen as 
$\tau_i^{PL} = 3.5$ and $\tau_i^L = 4.5$, respectively.
To evaluate the algorithmic performance we use two merit functions.  One measures  the distance from stationarity
of the average of the agents' iterates $\avg^t$ (cf. Th. \ref{thm:convergence}), and is defined as 
\begin{align*}
	J^t \triangleq 
	\Big \| \avg^t - \mathcal{P}_{\KK}\big(\SS_{\eta\lambda} 
	\big( \avg^t  - (\textstyle \sum \limits_{i=1}^N \nabla f_i (\avg^t ) - g(\avg^t) )\big) \big)
	\Big \|_\infty.
\end{align*}
Note that $J^t$ is a valid merit function: it is continuous and it is zero if and only if its argument  is a stationary 
solution of Problem~\eqref{regression}. 
The second merit function quantifies the    consensus disagreement at each iteration, and is defined as 
 \begin{align*}
  D^t \triangleq \max_{i\in\until{N}} \| \bx_{(i)}^t - \avg^t \|.
\end{align*}

The performance  of \algacronym/-PL and \algacronym/-L for different choices of the block dimension are reported 
in Figure~\ref{fig:convergence_disagreement_rate_PL} and Figure~\ref{fig:convergence_disagreement_rate_L}, respectively. 
{Recalling that $t$ is the iteration counter used in the
algorithm description, to fairly compare the algorithms' runs for different block
sizes, we plot $J^t$ and $D^t$, versus the normalized number of iterations $t/B$.}
The figures show that both  consensus and stationarity 
are achieved by \algacronym/-PL and \algacronym/-L within  $200$ message exchanges while  D-Grad  lacks behind . 

 {
Let $t_{\text{end}}$ be the completion time up to a tolerance of $10^{-4}$, i.e.,
  the number of iterations $t_{\text{end}}$ of the   algorithm such that
$J^{t_{\text{end}}} <10^{-4}$. Fig.~\ref{fig:blk_tradeoff} shows   $t_{\text{end}} / B$ versus the number of blocks $B$,
for \algacronym/-PL and \algacronym/-L.
The figure shows   that the communication cost reduces by increasing the number of
blocks, validating thus proposed block optimization/communication strategy. 
Note also that \algacronym/-PL outperforms \algacronym/-L.
 This is due to the fact that \algacronym/-PL better preserves the partial 
convexity of the objective function.
}
\begin{figure}[!htbp]
  \centering
	\includegraphics[scale=1]{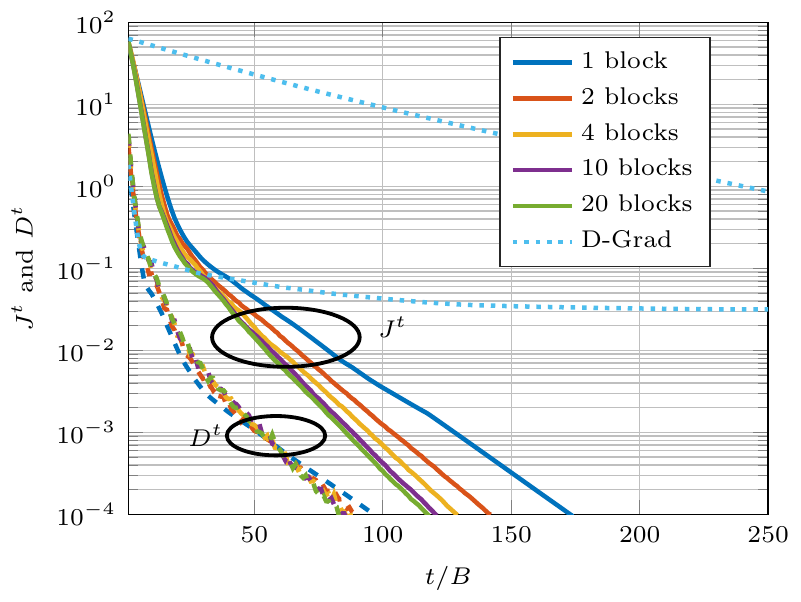}
\caption{
 Distance from stationarity $J^t$ (solid) and consensus disagreement  $D^t$ (dashed)
  versus the normalized number of iterations for several choices of the blocks' 
  number $B$ of \algacronym/-PL.
  }
\label{fig:convergence_disagreement_rate_PL}
\end{figure}

\begin{figure}[!htbp]
  \centering
	\includegraphics[scale=1]{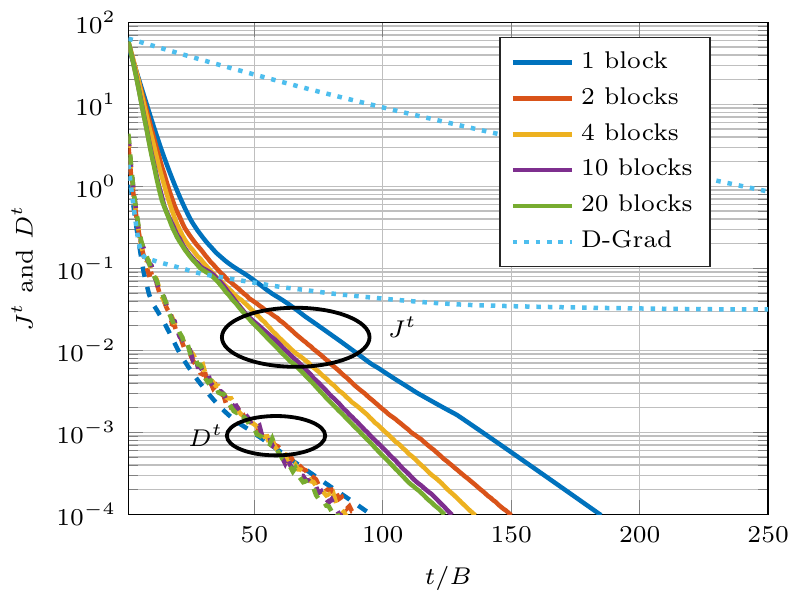}
\caption{
 Distance from stationarity $J^t$ (solid) and consensus disagreement  $D^t$ (dashed)
  versus the normalized number of iterations for several choices of the blocks' 
  number $B$ of \algacronym/-L.
  }
\label{fig:convergence_disagreement_rate_L}
\end{figure}

\begin{figure}[!htbp]
  \centering
	\includegraphics[scale=1]{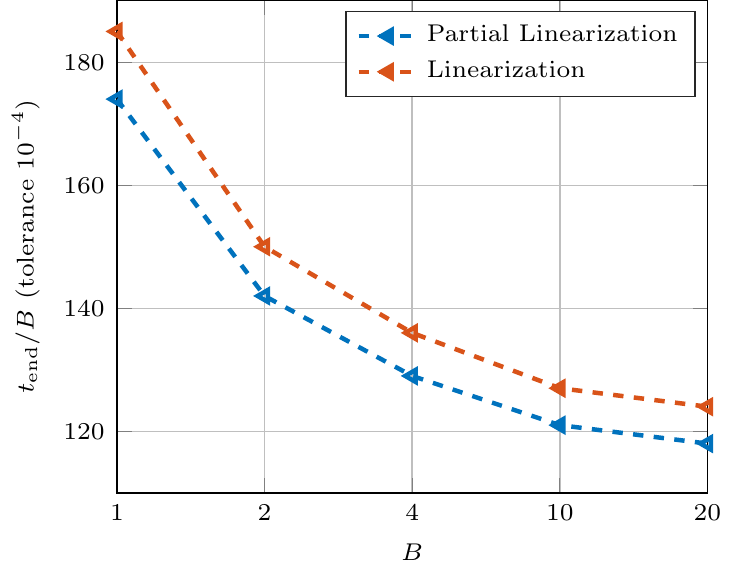}
\caption{
  Normalized completion time  to obtain $J^t <10^{-4}$ versus the number of blocks $B$,
  for \algacronym/-PL and \algacronym/-L.
  }
\label{fig:blk_tradeoff}
\end{figure}


\section{Conclusions}
\label{sec:conclusions}
In this paper we  proposed a novel block-iterative distributed scheme for
nonconvex, big-data optimization problems over (directed) networks. The key  novel feature of the scheme is a block-wise minimization from the agents of a convex approximation of the sum-utility, coupled with  a   block-wise consensus/tracking mechanism  aiming to average both the local copies of the agents' decision variables and the local estimates of the cost-function gradient.  Asymptotic convergence to a stationary solution of the problem  as well as consensus of the agents' local variables  was proved. 

\bibliographystyle{IEEEtran}
\bibliography{distributed_nonconvex_blocks}

\end{document}